\documentclass[letterpaper,10pt, conference]{ieeeconf} 

\usepackage{amsmath,amssymb,amsthm,epsfig,color,subfigure,empheq,graphicx}
\usepackage[algo2e]{algorithm2e}
\usepackage{enumerate,url,wasysym,balance,enumerate}
\usepackage{algorithm,algpseudocode}	
\usepackage{arydshln}

\usepackage{tikz,balance}
\usetikzlibrary{matrix,decorations.pathreplacing}

\newcommand \ba{\mathbf{a}}

\newcommand \bef{\mathbf{f}}

\newcommand \bi{\mathbf{i}}

\newcommand \bs{\mathbf{s}}
\newcommand \bu{\mathbf{u}}
\newcommand \bv{\mathbf{v}}
\newcommand \bw{\mathbf{w}}

\newcommand \by{\mathbf{y}}

\newcommand \bp{\mathbf{p}}
\newcommand \bq{\mathbf{q}}
\newcommand \bA{\mathbf{A}}

\newcommand \1{\mathbf{1}}

\newcommand \bxi{\boldsymbol{\xi}}

\newcommand \bI{\mathbf{I}}
\newcommand \bJ{\mathbf{J}}

\newcommand \bX{\mathbf{X}}
\newcommand \bY{\mathbf{Y}}
\newcommand \bZ{\mathbf{Z}}
\newcommand \bR{\mathbf{R}}

\newcommand \mcB{\mathcal{B}}
\newcommand \mcC{\mathcal{C}}
\newcommand \mcD{\mathcal{D}}
\newcommand \mcG{\mathcal{G}}

\newcommand \mcL{\mathcal{L}}

\newcommand \mcN{\mathcal{N}}

\newcommand \mcE{\mathcal{E}}

\newcommand \0{\boldsymbol{0}}
\newcommand \bphi{\boldsymbol{\phi}}

\def\realnumbers{\mathbb{R}}
\def\complexnumbers{\mathbb{C}}

\renewcommand{\baselinestretch}{0.97}

\DeclareMathOperator{\diag}{diag}

\newtheorem{proposition}{Proposition}
\newtheorem{lemma}{Lemma}

\theoremstyle{definition}
\newtheorem{remark}{Remark}

\newcommand{\longthmtitle}[1]{\mbox{}{\textit{(#1):}}}

\newcommand{\oprocendsymbol}{\hbox{$\bullet$}}
\newcommand{\oprocend}{\relax\ifmmode\else\unskip\hfill\fi\oprocendsymbol}

\allowdisplaybreaks

\usepackage{ifthen}
\newboolean{showcomments}
\setboolean{showcomments}{true}
\newcommand{\guido}[1]{\ifthenelse{\boolean{showcomments}}
	{ \textcolor{red}{(Guido says:  #1)}}{}}
\newcommand{\manish}[1]{\ifthenelse{\boolean{showcomments}}
	{ \textcolor{blue}{(Manish says:  #1)}}{}}

\IEEEoverridecommandlockouts

\title{\LARGE \bf Learning Local Volt/Var Controllers Towards Efficient Network Operation with Stability Guarantees}

\author{Guido Cavraro \and Zhenyi Yuan \and Manish K. Singh \and Jorge Cort\'es
\thanks{This work was authored by the National Renewable Energy Laboratory, operated by Alliance for Sustainable Energy, LLC, for the U.S. Department of Energy (DOE) under Contract No. DE-AC36-08GO28308. This work was supported by the Laboratory Directed Research and Development (LDRD) Program at NREL. The views expressed in the article do not necessarily represent the views of the DOE or the U.S. Government. The U.S. Government retains and the publisher, by accepting the article for publication, acknowledges that the U.S. Government retains a nonexclusive, paid-up, irrevocable, worldwide license to publish or reproduce the published form of this work, or allow others to do so, for U.S. Government purposes.}
\thanks{G. Cavraro is with the National Renewable Energy Laboratory. {\tt guido.cavraro@nrel.gov}. Z. Yuan and J. Cort\'es are with the Department of Mechanical and Aerospace Engineering, University of California, San Diego. {\tt \{z7yuan,cortes\}@ucsd.edu}. M. K. Singh is with the Department of Electrical and Computer Engineering, University of Minnesota. {\tt msingh@umn.edu.}}
}

\begin{document}
\maketitle

%

%
	
\begin{abstract} 
This paper considers the problem of voltage regulation in distribution networks. The primary motivation is to keep voltages within preassigned operating limits by commanding the reactive power output of distributed energy resources (DERs) deployed in the grid.
We develop a framework for developing local Volt/Var control that comprises two main steps. In the first, by exploiting historical data and for each DER, we learn a function representing the desirable equilibrium points for the power network. These points approximate solutions of an Optimal Power Flow (OPF) problem. In the second, we propose a control scheme for steering the network towards these favorable configurations. Theoretical conditions are derived to formally guarantee the stability of the developed control scheme, and numerical simulations illustrate the effectiveness of the proposed approach.
\end{abstract}

\section{Introduction}
%
%

The deployment of a massive number of DERs in distribution networks (DNs) is dramatically changing the electric power grid. Primarily driven by sustainability and economic incentives, DERs present additional opportunities, including voltage profile improvements and line-loss reduction. At the same time, the DERs' uncoordinated power injections or sudden generation changes could pose challenges to system stability and power quality.
To facilitate their integration in to power grids, DERs are being provided with sensing and computational capabilities and hence are becoming \emph{smart agents}. Further, they can exploit the flexibility of their power electronic interfaces to control the reactive power injection/withdrawal. Motivated by these observations, this paper aims to develop reactive power controllers to regulate voltages, also known as \emph{Volt/Var} controller, for DNs.

\emph{Literature Review:} Most control methods developed for DNs in recent years fit in the categories of \emph{distributed} or \emph{local} control strategies. In the first, DERs are allowed to communicate and share information in a communication network;
in the second, generators use only locally available information.
Distributed algorithms steer the network toward solutions of OPF problems, in which the power generation cost, the line losses, or the deviations from the nominal voltage are optimized~\cite{DallAnese_2018_TSG,Bolognani_2019_TCONES}.
Nevertheless, distributed strategies usually have precise and strict requirements on the communication network. For instance, in many works, each generator is required to share information with all its neighbors in the power system.
In local schemes, power injections are adjusted based on measurements taken at the point of connection of the power inverter to the grid~\cite{Turitsyn_2011_IEEE,Zhou_2021_TAC,Cavraro_2015_TPS}. The goal is typically to keep the voltages within safe limits.
Though simpler than distributed strategies, local schemes have intrinsic performance limitations, e.g., they might fail to regulate voltages even if the overall generation resources are enough~\cite{Bolognani_2019_TCONES}. 

To enhance the performance of local schemes and to reduce the gap with distributed and/or optimal controllers, recent efforts have devised customized control rules using data-driven and machine learning methods. A data set for learning control functions can be created by solving OPF problems using historical consumption and generation data, e.g., smart meter data.
Indeed, learning techniques have also been used to obtain fast (approximate) solutions to OPF problems. Deep neural networks (DNNs) have been employed to predict OPF solutions that are converted to a physically implementable schedule upon projection using a power flow solver~\cite{Zamzam_SGC_2020}. A graph neural network leveraging the connectivity of the power system is trained to infer AC-OPF solutions in~\cite{Owerko_2020_ICASSP}.
In~\cite{Singh_2020_SGC}, \cite{L2O} a DNN is trained to fit not only OPF minimizers, but also their sensitivities with respect to the problem inputs.
Piecewise linear control functions are designed in~\cite{Karagiannopoulos_2019_TSG} given the number of break points. 
The Authors in~\cite{cui_2021_arxiv} consider an OPF problem whose objective function penalizes the voltage deviations from the nominal one and the control effort. They derive stable local controllers that steer the system toward an approximated solution.
Continuous time local reactive power control schemes are designed in~\cite{shi_2021_arxiv} to solve an OPF problem with voltage constraints. However, reactive power capacity limits, which are critical when dealing with small-size generators, are not imposed.

\emph{Statement of Contributions:} In this work, we devise a framework for designing local Volt/Var scheme whose goal is to not only regulate voltages but also act as local surrogates of OPF solvers.
We advocate for a two-stage strategy. First, for each agent, a function, referred to as an \emph{equilibrium function}, providing OPF solution surrogates is learned from historical data. Precisely, such a function receives as input the local voltage and gives as an output an approximation of the optimal reactive power set point. Second, we devise a control algorithm whose equilibrium points (i) are asymptotically stable and (ii) are exactly the OPF approximated solutions provided by the equilibrium function.

The paper is structured as follows.
In Section~\ref{sec:modeling}, we model a power distribution network and define the OPF problem of interest.
Section~\ref{sec:approach} describes the aforesaid two-stage approach and Sections~\ref{sec:learn_eq_function} and~\ref{sec:control_rule} detail the corresponding technical results for each stage,respectively.
Finally, numerical tests are reported in Section~\ref{sec:tests}, and conclusions are drawn in Section~\ref{sec:conc}.

\emph{Notation:} Lower- (upper-) case boldface letters denote column vectors (matrices). Given a vector $\ba$, its $n$-th entry is denoted as $a_n$. Sets are represented by calligraphic symbols. The symbol $^\top$ stands for transposition, and inequalities are understood element-wise. The vector of all ones is denoted by $\1$; the corresponding dimension should be clear from the context. The operator $|\cdot|$ yields: the absolute value for real-valued arguments; the magnitude for complex-valued arguments; and the cardinality when the argument is a set. 
The set of complex numbers, of real numbers, and of nonnegative real numbers are denoted as $\complexnumbers,\realnumbers$, and $\realnumbers_{\geq 0}$, respectively. Operators $\Re(\cdot)$ and $\Im(\cdot)$ extract the real and imaginary parts of a complex-valued argument, respectively, and act entry-wise. Given a matrix $\bA$, an eigenvalue $\lambda$ with its associated eigenvector $\bxi$ forms the eigenpair $(\lambda, \bxi)$.
The norm of $\bA$ is defined by $\|\bA\| = \sqrt{\lambda_{\max}(\bA^\top \bA)}$, where $\lambda_{\max}(\bA^\top \bA)$ is the largest eigenvalue of $\bA^\top \bA$. This definition coincides with the 2-norm of a matrix.  The graph of a function $\phi: \realnumbers \rightarrow \realnumbers$ is the set of all points of the form $(x, \phi(x))$, with $x \in \realnumbers$.


\section{Power Distribution Grid Model}\label{sec:modeling}

Consider a power distribution network with $N+1$ buses  modeled by an undirected graph $\mcG = (\mcN, \mcE)$, whose nodes $\mcN = \{0, 1, \dots, N\}$ are associated with the electrical buses and whose edges represent the electric lines. We label the substation node as 0, and we assume that behaves as an ideal voltage generator imposing the nominal voltage of 1 p.u.
Define the following quantities:
\begin{itemize}
\item $u_n\in \complexnumbers$ is the voltage at bus $n \in \mcN$.
\item $v_n\in \realnumbers$ is the voltage magnitude at bus $n \in \mcN$.
\item $i_n\in \complexnumbers$ is the current injected at bus $n \in \mcN$.
\item $s_n=p_n+iq_n \in \complexnumbers$ is the nodal complex power at bus $n\in \mcN$, where $p_n,q_n\in \realnumbers$ are the active and  the reactive powers. Powers will take positive (negative) values, i.e., $p_n, q_n \geq 0$ ($p_n, q_n \leq 0$), when they are \emph{injected into} (\emph{absorbed from}) the grid.
\item $y_{(v,w)}  \in \complexnumbers$ is the admittance of line $(v,w) \in \mcE$.
\end{itemize}
Vectors $\bu,\bi,\bs\in \complexnumbers^n$ collect the complex voltages,  currents, and complex powers of buses $1,2,\ldots,n$; and the vectors $ \bv,\bp, \bq\in\realnumbers^n$ collect the voltage magnitudes, and their active and reactive power injections.
Denote by $z_e$
and by $y_e = z_e^{-1}$ the impedance and the admittance of line $e = (m,n) \in \mcE$.
The network bus admittance matrix $\bY \in \complexnumbers^{(N+1)\times(N+1)}$ is a symmetric matrix that can be expressed as $\bY = \bY_L + \diag(\by_T)$, where  
\begin{equation}
(\bY_L)_{mn} = \begin{cases}
- y_{(m,n)} & \text{ if } (m,n) \in \mcE, m \neq n, \\
\sum_{m \neq n} y_{(m,n)} & \text{ if }m = n,
\end{cases}
\label{eq:busAdmMatrixDef}
\end{equation}
and the vector $\by_T$ collects the shunt components of each line. The matrix $\bY_L$ is a complex Laplacian matrix, and hence satisfies $ \bY_L \1 = \0$.
We partition the bus admittance matrix separating the components associated with the substation and the ones associated with the other nodes, obtaining
$$\bY = \begin{bmatrix}
y_{0}&\by_0^\top \\
\by_0& \tilde \bY
\end{bmatrix}$$
with $y_{0} \in \complexnumbers, \by_0 \in \complexnumbers^{N}, \tilde \bY \in \complexnumbers^{N \times N}$.
If the network is connected, $\tilde \bY$ is invertible~\cite{Kettner_2018_TPS}. Let $\tilde \bZ := \tilde \bY^{-1}$, $\tilde \bR:=\Re\{\tilde \bZ\}$, and $\tilde \bX:=\Im\{\tilde \bZ\} \in \complexnumbers^{N \times N}$.
The power flow equation can be written as
\begin{subequations}
\begin{align}
& \bu = \tilde \bZ \bi+ \hat \bu , \label{eq:nodevoltage}\\
& u_0 = 1, \label{eq:PCCidealvoltgen}\\
& u_n \bar i_n = p_n + j q_n, \qquad n\neq 0 , \label{eq:nodeconstpwr}
\end{align}
\label{eq:nonlinvolt}
\end{subequations}
where $\bar i_n$ denotes the complex conjugate of $i_n$ and $\hat \bu := \tilde \bZ \by_0$. Equation~\eqref{eq:nodevoltage} represents the Kirchoff equations and provides the relation between voltages and currents. Finally, equation~\eqref{eq:nodeconstpwr} comes from the fact that all the nodes, except the substation, are modeled to be constant power buses.
Voltage magnitudes are nonlinear functions of the nodal power injections; however, using a first-order Taylor expansion, the power flow equation can be linearized to obtain
\begin{equation}
\bv = \tilde \bR \bp + \tilde \bX \bq + |\hat \bu|,
\label{eq:v=Rp+Xq}
\end{equation}
%
and to express the power losses as a scalar quadratic function of the power injections~\cite{Cavraro_2022_TCNS}
\begin{equation}
l = \bq^\top \tilde \bR \bq + \bp^\top \tilde \bR \bp.    
\label{eq:loss}
\end{equation}
%
%
%

Assume a subset $\mcC \subseteq \mcN$ of buses host DERs, with $|\mcC| = C$.
%
%
The remaining nodes constitute the set $\mcL=\mcN\setminus\mcC$. Every DER corresponds to a smart agent that  measures its voltage magnitude and performs reactive power compensation.
It is convenient to partition the reactive powers and voltage magnitudes by grouping together the nodes belonging to the same set
\begin{equation*}
\bq = \begin{bmatrix}
\bq_C^\top & \bq_L^\top
\end{bmatrix}^\top, 
\bv = \begin{bmatrix}
\bv_C^\top & \bv_L^\top
\end{bmatrix}^\top .
\label{eq:partitionV}
\end{equation*}
Also, the matrices $\tilde \bR$ and $\tilde \bX$ can be decomposed according to the former partition, yielding
\begin{equation}
\tilde \bR = \begin{bmatrix}

\bR & \bR_L \\
\bR_L^\top & \bR_{LL}
\end{bmatrix}, \quad
\tilde \bX = \begin{bmatrix}
\bX & \bX_L \\
\bX_L^\top & \bX_{LL}
\end{bmatrix}.
\label{eq:partitionRX}
\end{equation}
with $\bR$ and $\bX$ being positive-definite matrices.
Fixing the active and reactive loads along with the active solar generation, from \eqref{eq:v=Rp+Xq} and \eqref{eq:loss}, the voltage magnitudes and power losses become functions exclusively of $\bq_C$:
\begin{subequations}\label{eq:vl-approx}
    \begin{align}
    &\bv(\bq_C) = 
    \begin{bmatrix}
    \bX \\
    \bX_L^\top
    \end{bmatrix} \bq_C + \hat \bv \label{eq:appx_volt}\\
    & l(\bq_C) = \bq_C^\top \bR \bq_C + \bq_C^\top \bw + \hat l,  \label{eq:appx_loss}
\end{align}
\end{subequations}
where the following definitions are used 
\begin{subequations}\label{eq:defs}
\begin{align}
    \hat \bv &:= 
\begin{bmatrix}
\bX_L\\
\bX_{LL}
\end{bmatrix} \bq_L + \tilde\bR \bp + |\hat \bu|,\label{eq:barv-def}\\
\bw &:=2\bR_L\bq_L,\label{eq:w-def}\\
\hat l &:=\bq_L^\top\bR_{LL}\bq_L + \bp^\top \tilde \bR \bp.\label{eq:barl-def}
\end{align}
\end{subequations}

\section{Overview of the Proposed Approach for DER Control}\label{sec:approach}
%
%
This section proposes a two-stage approach to optimally use the flexibility in the DERs' reactive powers while ensuring the stable operation of the distribution network. In the first stage, a centralized OPF instance is formulated to determine the optimal DER reactive power set points given the noncontrollable (re)active power injections across the network. Altough the considered OPF formulation is convex, solving numerous instances of it for real-time operation might be computationally challenging. Further, the necessity for (re)active power information from across the network introduces communication challenges. Towards alleviating the aforementioned concerns,
%
%
we train a fleet of neural networks (one per DER)
%
%
to (approximately) predict the optimal set points, given merely local nodal voltages as inputs. In the second stage, we develop a control scheme to steer the DERs' reactive power injections to the set points obtained from the neural network outputs while formally guaranteeing stability.
%
%

A typical OPF formulation for the DERs' dispatch would solve for an optimal $\bq_C^*$, given the tuple $(\bp,\bq_L)$, such that the stipulated voltage limits and DER reactive-power capacity limits are satisfied, and a certain network criterion is optimized. Although arbitrary cost functions could be considered, here we consider an OPF problem that minimizes line losses. Such an OPF can be posed as
\begin{subequations}
	\begin{align}
	\bq_C^*(\bp,\bq_L):=\arg\min_{\bq_C}\  & ~l(\bq_C) \tag{P1}\label{eq:opf}\\
	\mathrm{s.t.}\  &\eqref{eq:vl-approx}-\eqref{eq:defs}, \text{and}\\
	&~\bv_{\min} \leq \bv(\bq_C) \leq \bv_{\max}, \label{eq:opf:c1}\\
	                &~\bq_{\min} \leq \bq_C \leq \bq_{\max}, \label{eq:opf:c2}
	\end{align} 
\end{subequations}
where $\bv_{\min}, \bv_{\max} \in \realnumbers^N$ are the desired voltage lower and upper limits on \emph{all} the network buses, and $\bq_{\min}, \bq_{\max} \in \realnumbers^C$ are the minimum and the maximum DERs' reactive power injections.
We denote the set of the feasible reactive power injections for the DER at node $n$
%
%
as $\mcB_n = \{q_n: q_n \in [q_{\min,n}, q_{\max,n}]\}$.
Problem~\eqref{eq:opf} is strictly convex, cf. \eqref{eq:appx_volt}--\eqref{eq:appx_loss}, and admits a unique minimizer.
Moreover, the minimizer is a function of the uncontrolled variables $\bp$ and $\bq_L$, which appear implicitly in the objective function and the constraint~\eqref{eq:opf:c1} via \eqref{eq:defs}. 

In principle, solving \eqref{eq:opf} given a tuple $(\bp,\bq_L)$ is tractable, thanks to the problem convexity. However, due to thehigh penetration of renewable generation, DNs are witnessing increased variability that requires solving numerous instances of \eqref{eq:opf} with a limited time and budget. To tackle this challenge, several neural network-based approaches have been put forth to predict approximates of $\bq_C^*$ with $(\bp,\bq_L)$  presented as the neural network inputs~\cite{Singh_2020_SGC}. Once trained, the time required for neural network inference when presented with a new input is minimal.  
%
%
While this alleviates the computational burden of solving OPFs, the need for the network-wide quantities $(\bp,\bq_L)$ imposes a significant communication burden for implementation. To simultaneously reduce  the computational and communication complexities, a common approach is to deploy solutions based on local control rules, whose performance in terms of optimality is generally lacking. For DER reactive power dispatches to achieve voltage regulation, such rules~\cite{IEEE1547} are often presented as piecewise linear functions of local voltages. Designing these rules to harness efficient DN operation has recently garnered tremendous interest~\cite{Zhou_2021_TAC,cui_2021_arxiv,Singh_GM_22}.

Inspired by the recently reported success of neural-network-based surrogates for OPF and ongoing efforts towards designing local control rules for DERs, this work proposes a two-stage approach. In the first stage, termed the \emph{learning stage}, we use historical data to learn functions that map voltages to (approximate) solutions of the OPF problem~\eqref{eq:opf}. Specifically, for each agent $n \in \mcC$, we aim to learn a function $\phi_n$ of the local voltage $v_n$ as
\begin{equation}
\phi_n:\realnumbers \rightarrow \mcB_n, \quad v_n \mapsto \phi_n(v_n),
\label{eq:loc_surr}
\end{equation}
with $\phi_n(v_n)$ providing the optimal reactive power surrogates. %
%
Then, we would like the generators to inject reactive power set points $\bq_C$ such that, for each $n\in \mcC$ we have
\begin{equation}
  q_n = \phi_n(v_n),
  \label{eq:fixed_point}
\end{equation}
where the voltage $v_n$ in turn depends on the reactive power injection $\bq_C$ as per~\eqref{eq:appx_volt}. Hence, the graph of the function $\phi_n$, namely, points of the form $(v_n,\phi_n(v_n))$, consists of desirable network configurations that are surrogates of the solutions of~\eqref{eq:opf} and, for this reason, we term the the functions $\{\phi_n\}_{n \in \mcC}$ \emph{equilibrium functions}. The second stage, termed the \emph{control stage}, aims to design local control rules that steer the network to configurations satisfying~\eqref{eq:fixed_point} for each $n\in \mcC$.

\begin{remark}\longthmtitle{On the need for a control algorithm}
The outcome of the learning stage are functions that map local voltages to (approximated) optimal reactive power set points. Hence, one might ask why it is not enough just to apply those reactive power setpoints provided by the learning function. This is the approach taken in e.g.,~\cite{Zamzam_SGC_2020,Jalali_2020_TSG}. The main reason why not is because we are considering the case in which only a few power injections, i.e., the DERs, are controlled. Applying the OPF solution surrogates $\bq_C^\sharp = \bphi(\bv_C)$, computed using the voltages $\bv_C$, in general, could change the voltages to a new configuration $\bv_C(\bq^\sharp) \neq \bv_C$. That is, $(v_n(\bq_C), q^\sharp_n)$ belongs to graph of $\phi_n$, but $(v_n(\bq_C^\sharp), q^\sharp_n)$ does not. Hence the new configuration is not an approximated power flow solution. The control scheme we develop aims exactly at iteratively steering the systems toward configurations belonging to the graph of the equilibrium functions. \oprocend
\end{remark}


\section{Neural Network-based Surrogates for Equilibrium Functions}
\label{sec:learn_eq_function}
This section describes our approach to learn equilibrium functions for each agent in $\mcC$ that describe the solutions of~\eqref{eq:opf} as a function of the individual voltages. The labeled dataset required to accomplish the desired learning task is obtained as described next. Given that \eqref{eq:opf} takes $(\bp,\bq_L)$ as input, we first build a set $\{(\bp^k,\bq_L^k)\}_{k=1}^K$ of $K$ load-generation scenarios. One can  obtain the aforementioned scenarios via random sampling from assumed probability distributions, historical data, or from forecasted conditions for a look-ahead period. 
Next, the OPF~\eqref{eq:opf} is solved for the $K$ scenarios to obtain the corresponding minimizers $(\bv(\bq_C^*),\bq_C^*(\bp,\bq_L))$. The entries for these minimizers are then separated for each $n\in\mcC$ to obtain datasets of the form $\mcD_n=\{(v_{n,k}^*,q_{n,k}^*)\}_{k=1}^K$, where the parametric dependencies have been omitted for notational ease. Next, we seek to independently learn equilibrium functions, one per node in~$\mcC$, such that the elements of the respective sets $\mcD_n$ are close to the graphs of the learned functions; with proximity quantified in terms of the squared error. Specifically,
%
%
using the mean square error (MSE) metric, the learning task can be posed as 
\begin{align}\label{eq:learning_problem}
    \min_{\phi_n} \frac{1}{K}\sum_{k=1}^K |\phi_{n,k}(v_{n,k}^*) - q_{n,k}^*|^2.
\end{align}
%
%
In addition, we impose the following conditions on each $\phi_n$: it needs to be $\mathrm{C1)}$ differentiable, $\mathrm{C2)}$ nonincreasing, and $\mathrm{C3)}$ with range in $\mcB_n$. The motivation for these requirements will be clear later. Since we employ neural networks to construct the equilibrium functions, ensuring that $\mathrm{C1)}-\mathrm{C3)}$ are satisfied is facilitated by choosing activation functions such as sigmoids, tanh, and softsign. In the following, we train the equilibrium functions using a single layer neural network and, as activation functions, we choose
$$\sigma(x) = \frac{e^x - e^{-x}}{e^x + e^{-x}}.$$
The next result gives a parameterization for a function satisfying $\mathrm{C1)}-\mathrm{C3)}$ using a single hidden layer neural network.

\begin{lemma}\longthmtitle{Parameterization of neural network satisfying the desired requirements}\label{lem:monotone_NN}
Consider a neural network $\mathrm{NN}(x): \realnumbers \mapsto \realnumbers$ with one hidden layer of $H$ neurons, with output defined as
\begin{align}\label{eq:monotone_func}
    \mathrm{NN}(x) =& \sum_{h=1}^{H} w_h \sigma(x + b_h),
\end{align}
where, $\sigma(\cdot)$ is the tanh activation function and $(w_h,b_h)$ denote the weight and bias associated with the $h$-th neuron. If $w_h \leq 0$, for all $h$, then $\mathrm{NN}$ is continuous, differentiable, and nonincreasing. Further, if $\sum_{h=1}^H|w_h|\leq W$, then $\mathrm{NN}(x)\in [-W,W]$, for all $x \in \realnumbers$.
\end{lemma}
\begin{proof}
The continuity and differentiability of $\mathrm{NN}$ trivially stems from that of $\sigma$. To establish the nonincreasing property, we take the derivative to obtain
\begin{align*}
    \frac{d \mathrm{NN}(x)}{d x} = \sum_{h=1}^H w_h \frac{d \sigma(x+b_h)}{d x} \leq 0,
\end{align*}
where we use the fact that the derivative of tanh function is always positive and that $w_h\leq0$, for all $h$. Owing to the above non-increasing property, the supremum (infimum) of $\mathrm{NN}$ is attained for the limit $x\rightarrow -\infty$ ($x\rightarrow \infty$). Substituting $\lim_{x\rightarrow -\infty}\sigma(x)=-1$ in~\eqref{eq:monotone_func} provides $\lim_{x\rightarrow -\infty}\mathrm{NN}(x)=\sum_{h=1}^H|w_h|\leq W$, where $w_h\leq 0$ is used. Similarly, evaluating for the limiting case $x\rightarrow\infty$, one obtains $\mathrm{NN}(x)\in[-W,W]$, thus completing the proof.
\end{proof}

Lemma~\ref{lem:monotone_NN} means that we can find the desired equilibrium functions $\{\phi_n\}_{n \in \mcC}$ by training the parameters of neural networks defined by~\eqref{eq:monotone_func}. The requirement that the range of $\phi_n$ belongs to $\mcB_n$ is satisfied by selecting $W= \min \{|q_{\min,n}|, |q_{\max,n}|\}$.
%
%

\section{A Local Control Scheme to Reach Desirable Equilibria}\label{sec:control_rule}

In this section, we propose and analyze a local control scheme that aims to steer the system to configurations satisfying~\eqref{eq:fixed_point} and~\eqref{eq:appx_volt}.
For each $n \in \mcC$, consider the following reactive power update rule
\begin{equation}
    q_n(t+1) = q_n(t) + \epsilon (\phi_n(v_n(t)) - q_n(t)),
    \label{eq:bus_react_upd}
\end{equation}
%
%
where $v_n(t)$ is determined by \eqref{eq:appx_volt}, and $\epsilon$ is a suitable positive number with $0\leq \epsilon <1$. Notice that, if algorithm~\eqref{eq:bus_react_upd} is initialized at $q_n(0) \in \mcB_n$, then $q_n(t) \in \mcB$ for all $t=1,2,\dots$; indeed, the new reactive power set point is a convex combination of two numbers in $\mcB_n$.
Algorithm~\eqref{eq:bus_react_upd} is a generalized version of the local scheme proposed in~\cite{Cavraro_2015_TPS}, which, instead of the learned $\phi_n$'s, considers linear functions. The following result characterizes the convergence properties of~\eqref{eq:bus_react_upd}. 

\begin{proposition}\longthmtitle{Asymptotic stability of equilibrium points}\label{prop:convergence}
Let the functions $\phi_n$'s meet conditions $\mathrm{C1)}-\mathrm{C3)}$, and define 
$$
M = \max_{{n \in \mcC}} \left\{ \max_{{v \in \realnumbers}} 
\left |
\frac{d \phi_n}{d v}
\right |
\right\}.
$$
%
%
If the stepsize parameter $\epsilon > 0$ satisfies
\begin{equation}
\label{eq:conv_cond}
\epsilon \leq \min \left\{1,\frac{2}{(1 + \|\bX\| M)} \right\} ,
\end{equation}
then the equilibria of the control rule~\eqref{eq:bus_react_upd} are asymptotically stable.
%
%
Moreover, if $\bq^\sharp$ is an equilibrium point and $\bv^\sharp = \bv(\bq^\sharp)$ is its associated voltage, then $(v^\sharp_n,q^\sharp_n)$ belongs to the graph of $\phi_n$ for every $n\in\mcC$.
\end{proposition}
%
%

\begin{proof}
To prove Proposition~\ref{prop:convergence}, it is convenient to express~\eqref{eq:bus_react_upd} in vectorial form as
\begin{align}
    \bq_C(t+1) & = (1-\epsilon)\bq_C(t) + \epsilon \bphi(\bv_C(\bq(t))) \notag \\
    & = \bef(\bq_C(t)),
    \label{eq:vec_bus_react_upd}
\end{align}
where $\bphi: \realnumbers^C \rightarrow [\bq_{\min},\bq_{\max}]$ collects all the $\phi_n$'s, and $\bef$ is the operator
\begin{align*}
&\bef: [\bq_{\min},\bq_{\max}] \rightarrow [\bq_{\min},\bq_{\max}]  \\
& \bq_C \mapsto (1-\epsilon)\bq_C + \epsilon \bphi_C(\bv(\bq)).
\end{align*}
Using the chain rule and equation~\eqref{eq:appx_volt}, the Jacobian of $\bef$ can be expressed as
\begin{equation}
    \bJ_f = (1-\epsilon)\bI + \bJ_{\phi} \bX,
    \label{eq:Jacobian}
\end{equation}
where $\bJ_\phi$ is the Jacobian of $\bphi$ and can be explicitly written as
$$\bJ_\phi = \diag\left(\left\{\frac{d \phi_n(v_n)}{d v_n}\right\}\right).$$
Notice that $\bJ_\phi$ is a diagonal matrix with nonpositive entries, because of property (i). Hence,~\eqref{eq:Jacobian} can be rewritten as
\begin{equation}
    \bJ_f = (1-\epsilon)\bI - |\bJ_{\phi}| \bX.
    \label{eq:Jacobian2}
\end{equation}
Let $(\lambda_i, \bxi_i)$ be an eigenpair for $|\bJ_{\phi}| \bX$. Trivially, $(1 - \epsilon - \epsilon \lambda_i, \bxi_i)$ is an eigenpair for $\bJ_\phi$. Hence, for the asymptotic stability of the equilibrium points of~\eqref{eq:bus_react_upd}, we need to ensure that
$$|1 - \epsilon - \epsilon \lambda_i| < 1 $$
for any eigenvalue $\lambda_i$ of $|\bJ_{\phi}| \bX$. The former can be split into two inequalities. The first yields $\lambda_i > -1$, which is always true since $|\bJ_{\phi}| \bX$ is positive semidefinite. The second instead reads 
$\epsilon (1 + \lambda_i) < 2$ and, using Lemma~\ref{lem:eigvls} (in Appendix), always holds if $\epsilon (1 + \|\bX\| M) < 2$ or, equivalently if
$$\epsilon  < \frac{2}{(1 + \|\bX\| M)}.$$
Further, recall that the algorithm is defined for $0 < \epsilon < 1$. Equation~\eqref{eq:conv_cond} then follows.
Finally, if $\bq^\sharp$ is an equilibrium of~\eqref{eq:bus_react_upd}, by definition from equation~\eqref{eq:vec_bus_react_upd} we have
$$\bq^\sharp = (1-\epsilon)\bq^\sharp + \epsilon \bphi(\bv(\bq^\sharp)) $$
and thus
$$\bq^\sharp = \bphi(\bv(\bq^\sharp))$$
and $(v^\sharp_n,q^\sharp_n)$ belongs to the graph of $\phi_n$ for every $n\in\mcC$.
\end{proof}

\begin{remark}\longthmtitle{Interpretation of the requirements on the learned equilibrium functions}
We explain here the reasons for the requirements $\mathrm{C1)}-\mathrm{C3)}$ on the equilibrium functions $\{\phi_n\}_{n \in \mcC}$. Constraining the range of each $\phi_n$ to $\mcB_n$ ensures that the reactive power set points are always feasible and avoids the use of projections in~\eqref{eq:bus_react_upd}. The continuity, the differentiability, and the monotonicity assumptions are instead used in the proof of Proposition~\ref{prop:convergence}, i.e., these requirements on the learning of the equilibrium functions guarantee  the stability of the closed-loop system. This is done at the cost of potentially increasing the optimality gap.  \oprocend
%
\end{remark}
%
%

\begin{remark}\longthmtitle{Non-incremental vs. incremental control rules}
One could think to update the reactive power using the rule
\begin{equation}
    q_n(t+1) = \phi_n(v_n(t)),
    \label{eq:ni_bus_react_upd}
\end{equation}
where $v_n(t)$ is determined by~\eqref{eq:appx_volt}. Following~\cite{Farivar_2015_SGC}, we refer to algorithms like~\eqref{eq:ni_bus_react_upd} as \emph{non-incremental}, because the new set points are determined based on the local voltage without explicitly exploiting a memory of past set points. These approaches can thus result in large variations in reactive-power set points across time steps. Instead, we refer to algorithms like~\eqref{eq:bus_react_upd} as \emph{incremental} because they compute small (as determined by $\epsilon$) adjustments to the current set points. Current practice is indeed to update the reactive powers using non-incremental algorithms, e.g., see~\cite{Turitsyn_2011_IEEE} or the IEEE Std 1547~\cite{Std1547}. It is trivial to see that equilibrium points of~\eqref{eq:ni_bus_react_upd} belong to the graph of the equilibrium function, too.
The main issue is ensuring the convergence of~\eqref{eq:ni_bus_react_upd}: several works~\cite{Zhou_2021_TAC, Cavraro_2015_TPS} provide conditions that guarantee the stability of non-incremental algorithms, usually expressed as bounds on the voltage function slope.
Actually, one can show that~\eqref{eq:ni_bus_react_upd} converges if 
\begin{equation}
M \leq \frac{1}{\|\bX\|}.
\label{eq:static_conv}
\end{equation}
To use~\eqref{eq:ni_bus_react_upd}, one would then need to additionally enforce~\eqref{eq:static_conv}
%
%
in the learning process described in Section~\ref{sec:learn_eq_function}. The resulting equilibrium function would then provide approximations of the OPF solutions that are worsened because of the additional restriction.
%
%
By contrast, the incremental approach in~\eqref{eq:bus_react_upd} can handle arbitrary finite maximum slopes~$M$ by choosing a suitable stepsize $\epsilon$ that satisfies the condition~\eqref{eq:conv_cond}. \oprocend
\end{remark}

\section{Numerical tests}\label{sec:tests}
%
%

\begin{figure}[t]
\centering	
\includegraphics[width=0.7\columnwidth]{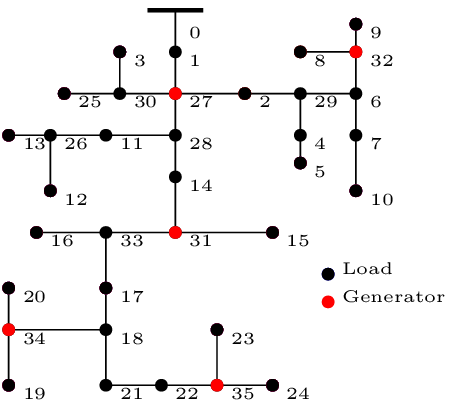}
\caption{The IEEE 37-bus feeder.}
\label{fig:ieee37}
\end{figure}
%
%

We conduct a case study on the IEEE 37-bus feeder upon removing the regulators, incorporating five solar generators, and converting it to its single-phase equivalent, see Fig.~\ref{fig:ieee37}. These five solar generators are the DERs that we intend to control.

\emph{Simulation setup.} We use the Matlab-based OPF solver Matpower~\cite{MATPOWER} to compute both the exact optimal solution of~\eqref{eq:opf} and the solution of the power flow equation. We implement the neural networks using TensorFlow 2.7.0 and conduct the training process in Google
Colab with a single TPU with 32 GB memory. The number of episodes and the number of neurons $H$ are 1000 and 200, respectively. The neural networks are trained with the learning rate set to 0.01 using the Adam optimizer~\cite{DPK-JB:15}.
%

%
%

\emph{Real-world dataset.} The feeder has 25 buses with non-zero load. We extract minute-based load and solar generation data for June 1, 2018, from the Pecan Street dataset~\cite{pecandata}, and the first 75 nonzero load buses from the dataset are aggregated every 3 loads and normalized to obtain 25 load profiles. Similarly, we obtain five solar generation profiles for the active power of DERs.
The normalized load profiles for the 24-hour period are scaled so that 97\% of the total load duration curve coincides with the total nominal load. This scaling results in a peak aggregate load being 1.1 times the total nominal load. We synthesize reactive loads by scaling active demand to match the power factors of the IEEE 37-bus feeder.
The 5 DERs have different generation capabilities, precisely, $\bq_{\max} = [0.4020 \ 0.4020 \ 0.4020 \ 0.0500 \ 0.0500]^\top$ and $\bq_{\min} = -\bq_{\max}$.
Voltage limits are set to $\bv_{\max} = 1.03$ p.u. and $\bv_{\min} = 0.97$ p.u.
Fig.~\ref{fig:solarload} shows the total demand and solar generation across the feeder.
Fig.~\ref{fig:equi_func} plots the learned equilibrium function of DER 31, along with the exact optimal reactive power set points obtained by solving by~\eqref{eq:opf}. 

\begin{figure}[htb]
    \centering
    \includegraphics[scale=0.2]{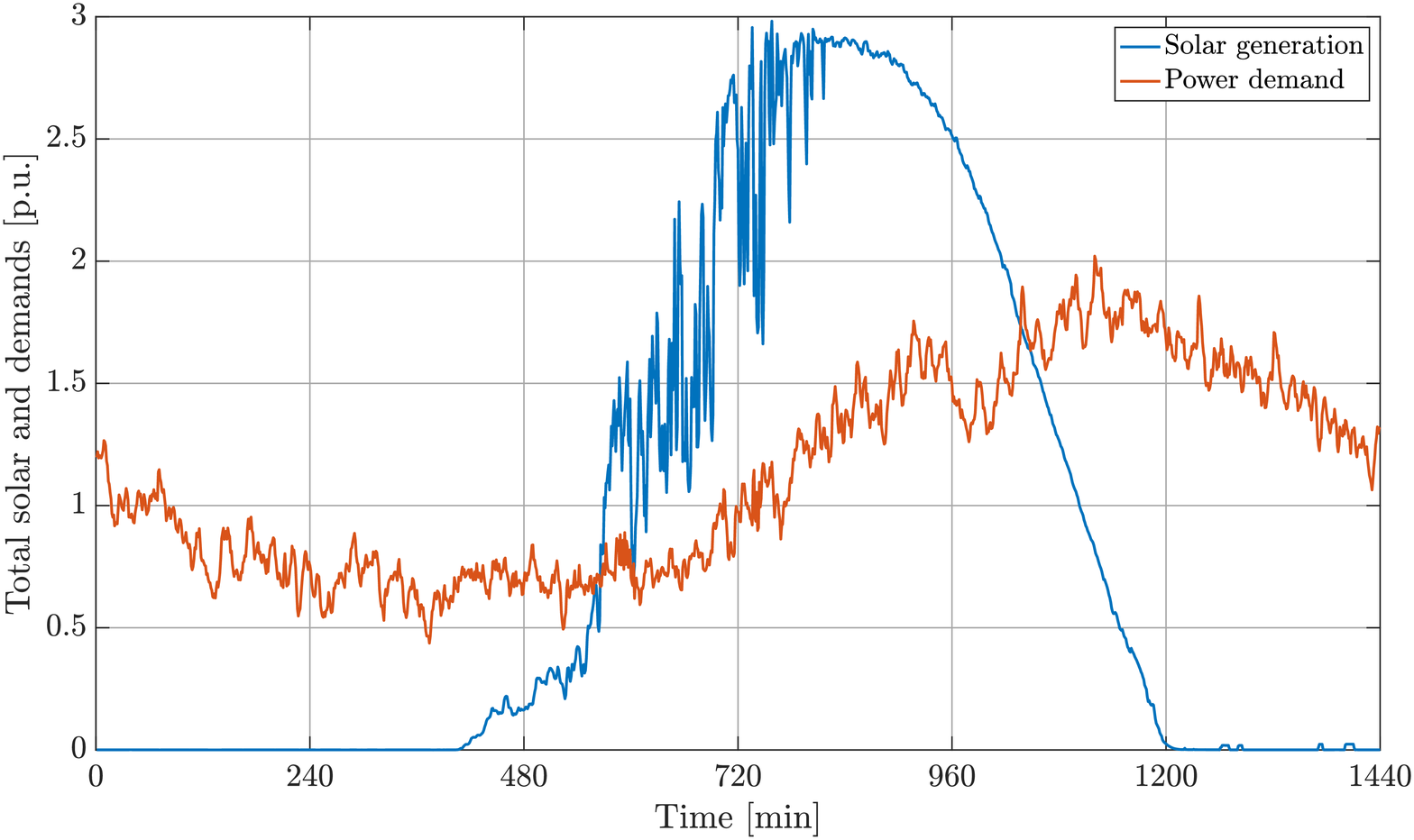}\\
    \caption{Minute-based data for the total (feeder-wise) solar power generation and active power demand.}
    \label{fig:solarload}
\end{figure}

\begin{figure}[htb]
    \centering
    \includegraphics[scale=0.2]{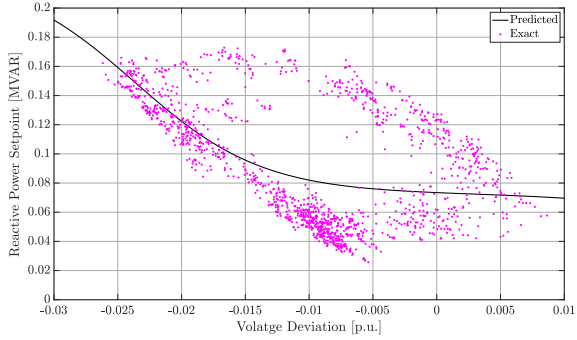}
    \caption{Learned equilibrium function for DER 31 along with the dataset points in purple.}
    \label{fig:equi_func}
\end{figure}
%
%

\emph{Simulation results.} We first verify the stability properties of the local control algorithm~\eqref{eq:bus_react_upd} stated in Proposition~\ref{prop:convergence}. Fig.~\ref{fig:convergence} reports the evolution of the DERs' reactive power injections when loads are fixed. The power trajectories converge to their final value.
Next, we run the control algorithm~\eqref{eq:bus_react_upd} in a scenario where loads are time-varying. Specifically, we obtain loads by randomly perturbing the consumption data used to learn the equilibrium functions. This can be interpreted as having the data from the dataset prescribing a day-ahead forecast, whereas their random perturbation act as the true realization of the load. These loads are minute-based and we consider 120 iterations of~\eqref{eq:bus_react_upd} per minute. We compare the performance of the system when the agents perform~\eqref{eq:bus_react_upd}
%
%
with the one where control actions are not taken.
%
%
Fig.~\ref{fig:volt_dev} reports the minimum voltage deviations, i.e., $\bv - \1$, and Fig.~\ref{fig:cost} the line power losses. In contrast to the uncontrolled case, our approach brings the voltages back to the desired voltage region, and significantly reduces line losses. 

\begin{figure}[htb]
\centering
\includegraphics[scale=0.2]{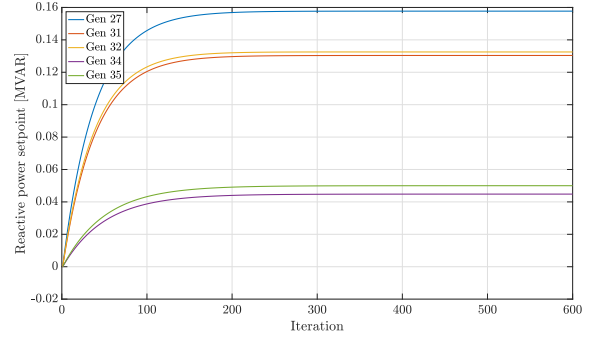}
\caption{The convergence property of the local control schemes, where we use the power data of the 1095-th minute and consider 600 iterations of \eqref{eq:bus_react_upd} with $\epsilon=0.01$.}
\label{fig:convergence}
\end{figure}

\begin{figure}[htb]
\centering
\includegraphics[scale=0.2]{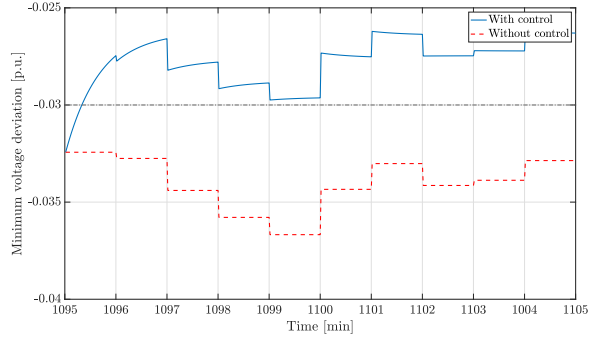}
\caption{Comparison of the minimum voltage deviations between the proposed approach and the uncontrolled case during time period $[1095,1105]$ minutes with 120 iterations of \eqref{eq:bus_react_upd} per minute and $\epsilon=0.01$.}
\label{fig:volt_dev}
\end{figure}

\begin{figure}[htb]
\centering
\includegraphics[scale=0.2]{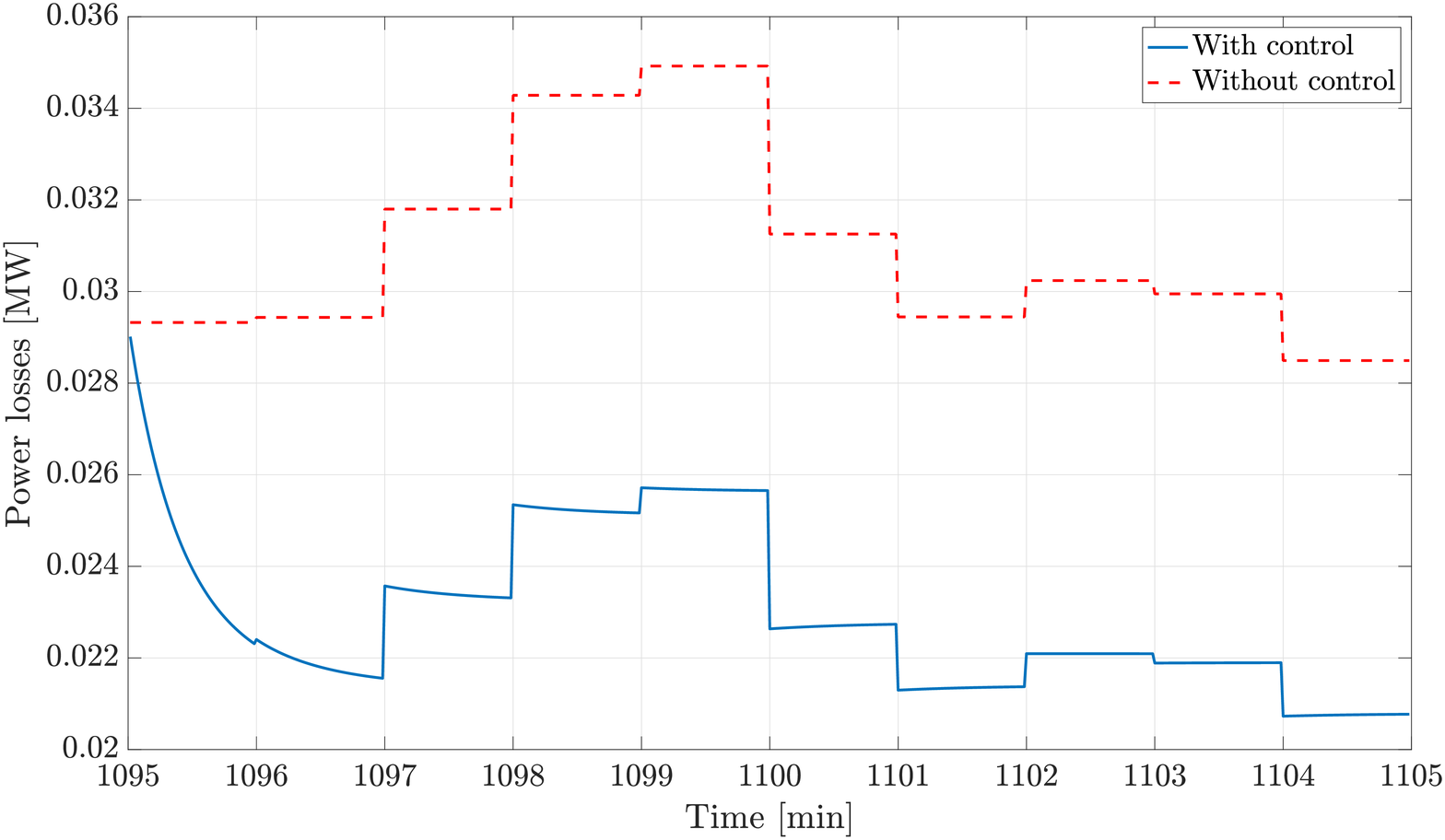}
\caption{Comparison of the power losses between the proposed approach and the uncontrolled case during time period $[1095,1105]$ minutes with 120 iterations of \eqref{eq:bus_react_upd} per minute and $\epsilon=0.01$.}
\label{fig:cost}
\end{figure}

%
%


\section{Conclusions}\label{sec:conc}

We have put forward a two-stage approach to the design of local volt/var control schemes capable of steering DNs toward desirable equilibria. In the first stage, we learn the equilibrium function for each DER bus that, given the local voltage, provides as an output a reactive power set point. Points in the graph of the equilibrium function represent approximations of solutions of an OPF problem. 
We employ a neural network representation that, by design, has the resulting equilibrium function be differentiable, nonincreasing (but without constraints on the slope),  and bounded. In the second stage, we devise an incremental control algorithm whose equilibria belong to the graph of the equilibrium function. The properties of the learned equilibrium maps play a key role in showing that the equilibria are asymptotically stable. Future research directions include reducing the optimality gap,  relaxing the differentiability requirement on the equilibrium maps, and extending the proposed framework to the more general scenario where communication among neighboring agents is allowed.

\appendix 
\begin{lemma}
\label{lem:eigvls}
The matrix $|\bJ_{\phi}| \bX$ is positive semidefinite. Moreover, if $\lambda_{\max}$ is its maximum eigenvalue, it holds
\begin{equation}
    \lambda_{\max} \leq \|\bX\| M.
    \label{eq:max_egvl}
\end{equation}
\end{lemma}
\begin{IEEEproof}
First, we show that $\bX$ is a positive definite matrix. Let $(\lambda_i, \bxi_i)$ be an eigenpair for $|\bJ_{\phi}| \bX$. Then, $(\lambda_i, \bX^{\frac 1 2}\bxi_i)$ is an eigenpair for the symmetric positive semidefinite matrix $\bX^{\frac 1 2}|\bJ_{\phi}| \bX^{\frac 1 2}$. Indeed,
\begin{align*}
   \bX^{\frac 1 2}|\bJ_{\phi}| \bX^{\frac 1 2} \bX^{\frac 1 2}\bxi_i &= \bX^{\frac 1 2}|\bJ_{\phi}| \bX \bxi_i 
   = \lambda_i \bX^{\frac 1 2} \bxi_i 
\end{align*}

Hence, $|\bJ_{\phi}| \bX$ is a positive semidefinite matrix, too.
Moreover, using the triangle inequality and because $\bJ_\phi$ is a diagonal matrix, we have that
\begin{align*}
    \lambda_{\max} = & \|\bX^{\frac 1 2}|\bJ_{\phi}| \bX^{\frac 1 2}\| \leq \|\bX\| \|(|\bJ_\phi|)\| \\
    & \leq \|\bX\| \max_{n \in \mcC}\left\{ \max_{v \in \realnumbers} \left\{\frac{d \phi_n(v)}{d v}\right\}\right\} = \|\bX\| M.
\end{align*}

\end{IEEEproof}

\balance
\bibliographystyle{IEEEtran}
\bibliography{myabrv,bibliography}	
\end{document}